\documentclass[a4paper,12pt]{article}

\usepackage{amsfonts}
\usepackage{amsthm}
\usepackage{amsmath}
\usepackage{amssymb}
\usepackage{color}
\usepackage{a4wide}

\newtheorem{Thm}{Theorem}[section]
\newtheorem{Lem}[Thm]{Lemma}
\newtheorem{Cor}[Thm]{Corollary}
\newtheorem{Prop}[Thm]{Proposition}
\newtheorem{Def}[Thm]{Definition}

\newcommand{\man}[1]{\ensuremath{\mathcal{#1}}} 
\newcommand{\bound}[2][]{\ensuremath{\partial({#1}(\man{#2}))}} 
\newcommand{\ABound}[1][M]{\ensuremath{\mathcal{B}(\mathcal{#1})}} 
\newcommand{\seq}[1]{\ensuremath{\mathfrak{#1}}} 
\newcommand{\Seqz}[1][\man{M}]{\ensuremath{\Sigma_0({#1})}}  
\newcommand{\Seqs}[2]{\ensuremath{\Sigma({#1},{#2})}} 
\newcommand{\covers}{\ensuremath{\rhd}} 
\newcommand{\EmptySet}{\ensuremath{\varnothing}}
\newcommand{\Cau}[1][d]{\ensuremath{\mathcal{C}({#1})}} 
\newcommand{\Dis}[1][\man{M}]{\ensuremath{D({#1})}} 

\title{A Correspondence Between Distances and Embeddings for Manifolds: New Techniques for Applications of the Abstract Boundary}

  \author{Ben E. Whale\footnote{Corresponding Author}\ \footnote{ben.whale@anu.edu.au}\ \footnote{Centre for Gravitational Physics,
College of Physical \& Mathematical Sciences, The Australian National University, Canberra, ACT 0200,
AUSTRALIA} , Susan M. Scott\footnote{susan.scott@anu.edu.au}\ \footnotemark[3]}

\begin{document}

  \maketitle

  \begin{abstract}
    We present a one-to-one correspondence between
    equivalence classes of embeddings of a manifold (into a larger manifold of the same
    dimension)
    and equivalence
    classes of certain distances on the manifold.  This
    correspondence allows us to use the Abstract Boundary to describe the structure of the
    `edge' of our manifold without resorting to structures
    external to the manifold itself.  This is particularly
    important in the study of singularities within General
    Relativity where singularities lie on this `edge'.  The
    ability to talk about the same objects, e.g., singularities, via different
    structures provides alternative routes for investigation which
    can be invaluable in the pursuit of physically motivated problems where certain types of information are unavailable or difficult to use.
  \end{abstract}


\section{Introduction}
  The study of singularities within General Relativity suffers
  from a unique problem in physics: there is no background metric
  in which the singularity exists.  Yet our intuition wishes to
  describe these `singularities' with a location and physical
  properties.  There are also the additional problems of providing
  a co-ordinate independent definition of a singularity and a
  description of the full range of singular behaviour.

  There are a number of boundary constructions which attempt
  to provide both a definition of and a location for singularities
  of space-times (see \cite{Ashley2002a} for a review of the most notable
  boundary constructions, \cite{Senovilla1998} for a review of the field in
  general and \cite{BeemEhrlichEasley1996} for proofs of the most important applications) for a review of the field in
  general).  The three most common examples are the $g$-boundary \cite{Geroch1968a},
  $b$-boundary \cite{Schmidt1971} and the $c$-boundary \cite{GerochPenroseKronheim1972}\footnote{See the preprint, \cite{Flores2010Final}, for an
  up to date review of the $c$-boundary including work conducted at the same time as this paper.}.  Each of these uses some
  aspects of the metric structure to identify `missing' points
  from the space-time and then prescribes a method to re-attach
  them.  Because each of these constructions uses the metric
  structure in this way, they each suffer from a variety of flaws.
  The 
  Abstract Boundary (or $a$-boundary)
  \cite{ScottSzekeres1994} avoids these flaws by only using topological information in its construction.
  
  The Abstract Boundary avoids using the metric structure by using embeddings as a kind of reference for boundary
  structure. Specifically, it uses the set of
  all embeddings $\phi:\man{M}\to\man{M}_\phi$, where $\man{M}$ is the manifold of
  our original space-time and $\man{M}_\phi$ is a manifold of the
  same dimension as $\man{M}$ (we shall refer to such an embedding as an envelopment), 
  to construct a set of equivalence
  classes of boundary points of these embeddings.  Each equivalence class corresponds to the
  representation of a `missing' point in a co-ordinate chart.
  Hence a `missing' point may look very different depending on the
  chosen chart.  Remarkably, a definition and location of a singularity can be
  retrieved from this rather general set-up.
  
  We note that there have been two important recent contributions in this area.
  The first is Garc\'{i}a-Parrado and Senovilla's isocausal boundary
  \cite{0264-9381-22-21-009,GarciaParrado:2002xt,0264-9381-22-9-R01} which is, in part, inspired from the Abstract Boundary. The
  second is a number of recent developments of the $c$-boundary; we refer the reader to S\'{a}nchez' interesting
  paper \cite{Sanchez2009e1744}. In both
  cases their work was more directly concerned with causal structures.

  The very nice thing about the $a$-boundary is that it avoids all the
  usual problems inherent in metrically constructed boundaries.  Unfortunately,
  this comes at a cost. In particular, complete knowledge of the $a$-boundary is reliant 
  on knowing all possible
  envelopments of \man{M}.  This reliance makes it almost impossible to
  construct the complete $a$-boundary for a general space-time.  It should be emphasised,
  however, that this does not, in any way, hinder its utility in the investigation of problems related to singularities  
  (e.g., \cite{Philpot2004}).
  In much the same way as one does not need to know all charts in the atlas of a manifold, so too one does not
  need the full Abstract Boundary to extract information about the `edge' of space-times. This is more a matter
  of representation than of missing information.

  This paper demonstrates that there is a one-to-one
  correspondence between the set of equivalence classes of
  envelopments and a set of equivalence classes of distances on the manifold \man{M}. We give a short example
  of how this correspondence can be used to investigate the $a$-boundary structure of space-times. In a future paper
  the authors will demonstrate that this correspondence can be used to construct the complete
  $a$-boundary from this set of equivalence classes of distances.
  Thus the correspondence provides an alternative method for studying the $a$-boundary.
  While this does not solve the problem mentioned above, it does
  make it more readily accessible.

  We also hope that this correspondence will be of interest to all
  mathematicians desiring to study the `edges' of manifolds. The work below demonstrates that the
  $a$-boundary has a strong relationship to Cauchy structures, in the sense of \cite{LowenColebunders1989}, and thus also
  to the more normal boundaries, e.g., the Stone-\u{C}ech compactification, employed in topology.
  
  Section \ref{Sec.Prelim} introduces the necessary 
  background for the Abstract Boundary.
  Sections \ref{sec.envelopments} and \ref{sec.distances} present equivalence 
  relations on the set of all envelopments of a manifold and a set of distances on a manifold, respectively, and discuss the relation of this work to
  Cauchy structures.  The first relation describes
  when two envelopments provide the same information about the $a$-boundary.  The second relation
  mirrors the ideas of the first, but on a set of distances rather than envelopments.  Section \ref{sec.correspondence}
  gives a one-to-one correspondence between the two sets of equivalence classes, thereby
  showing that what can be constructed using the first can also be constructed using the second. Section \ref{sec.demonstration}
  presents a short example of how this correspondence can be used.

  Our main result, contained in section \ref{sec.correspondence}, is that the set of equivalence classes of envelopments, relevant for the Abstract 
  Boundary, is in one-to-one
  correspondence with a set of equivalence classes of distances. So, in effect, the main 
  result states that, in order to study the Abstract Boundary, one can use either envelopments or a certain subset of distances.
  The ability to employ distances when using the Abstract Boundary to investigate problems should provide both greater flexibility
  and accessibility.
  
  To construct this correspondence we will use three homeomorphisms, between the closures, $\overline{\phi(\man{M})}$, 
  $\overline{\psi(\man{M})}$, of the images of $\man{M}$ under
  equivalent envelopments, $\phi,\psi$, between the Cauchy completion, $\man{M}^d,\man{M}^{d'}$, 
  of $\man{M}$ of equivalent distances, $d,d'$ and between the closure, $\overline{\phi(\man{M})}$, of the image of 
  $\man{M}$ under an envelopment $\phi$ and
  the Cauchy completion of $\man{M}$ with respect to a distance, $d_\phi$, that is related to the envelopment.
  The existence of these homeomorphisms is ensured by propositions \ref{DisEm:Prop.ConditionsForEquivalenceOfEmbeddings},
  \ref{DisEm:Prop.DistanceEquivalenceStatmentsNotInDis} and corollary \ref{CorSusan}. The propositions both show that certain functions have extensions into
  the completion of their domains. It is here
  that the theory of Cauchy spaces underlies our result as the functions we consider are not necessarily uniformly
  continuous and therefore their well known extension theorem does not apply. 
  The needed generalisation of this extension theorem
  is expressed in the language of Cauchy spaces; see subsections \ref{envelopments cauchy}
  and \ref{distance Cauchy}.

\section{Preliminary results and notation}\label{Sec.Prelim}

  We need a few results and definitions from previous papers about
  the Abstract Boundary; they are collected below for the
  convenience of the reader.
  We recommend that the reader refer to the cited papers for a
  detailed introduction to the subject.

  \begin{Def}[see \cite{ScottSzekeres1994}]\label{Def:Envelopment}
    Let $\man{M}$ and $\man{M}'$ be manifolds of the same
    dimension. If there exists $\phi:\man{M}\to\man{M}'$ a $C^{\infty}$
    embedding, then $\man{M}$ is said to be enveloped by
    $\man{M}'$, $\man{M}'$ is the enveloping manifold and $\phi$
    is an envelopment. Since both manifolds have the same dimension, $\phi(\man{M})$
    is open in $\man{M}'$.
  \end{Def}
  
  \begin{Def}[see \cite{ScottSzekeres1994}]
    Let $\phi:\man{M}\to\man{M}_\phi$ be an envelopment.  A non-empty subset $B$ of 
    $\partial(\phi(\man{M}))$ is called a boundary set.
  \end{Def}

  \begin{Def}[see \cite{ScottSzekeres1994}]
    A boundary set $B\subset\partial(\phi(\man{M}))$ is said to cover another boundary set 
    $B'\subset\partial(\psi(\man{M}))$ if and only if for every open neighbourhood $U$ of $B$ in
    $\man{M}_\phi$ there exists an open neighbourhood $V$ of $B'$ in $\man{M}_\psi$ so that
    \[
      \phi\circ\psi^{-1}(V\cap\psi(\man{M}))\subset U.
    \]
    
    We shall denote this partial order by $\covers$, so that $B$ covers $B'$ is written $B\covers B'$.  If a boundary
    set $B$ is such that it consists of a single point, i.e., $B=\{p\}$, then we shall simply write $p\covers B'$ rather than the
    more cumbersome $\{p\}\covers B'$.
  \end{Def}

  \begin{Thm}[see \cite{ScottSzekeres1994}]\label{Thm:Abstract Boundary Covering Relation -
      sequences}
    A boundary set $B_\phi\subset\partial(\phi(\man{M}))$ covers another
    boundary set $B_\varphi\subset\partial(\varphi(\man{M}))$ if and
    only if for every sequence $\{x_i\}\subset\man{M}$ so that
    $\{\varphi(x_i)\}$ has an accumulation point in
    $B_\varphi$,
    the sequence $\{\phi(x_i)\}$ has an accumulation point in
    $B_\phi$.
  \end{Thm}
  
  \begin{Def}[see \cite{ScottSzekeres1994}]
    Two boundary sets $B\subset\partial(\phi(\man{M}))$ and $B'\subset\partial(\psi(\man{M}))$ are equivalent if and only if
    $B\covers B'$ and $B'\covers B$.  We shall denote equivalence by $B\equiv B'$, and, as before, if $B=\{p\}$ then we shall simply write
    $p\equiv B'$ rather than $\{p\}\equiv B'$.
    
    We shall denote the equivalence class of the boundary set $B$ by $[B]$. That is, 
    $[B]=\{B'\subset\partial(\varphi(\man{M})): B'\equiv B\text{, where $\varphi$ is an envelopment of \man{M}}\}$.  As usual, if $B=\{p\}$, then
    we shall write $[p]$ rather than $[\{p\}]$.
  \end{Def}
  
  \begin{Def}[see \cite{ScottSzekeres1994}]
    The Abstract Boundary ($a$-boundary), $\ABound$, of a manifold $\man{M}$ is the set of all equivalence classes of boundary sets that contain a singleton, 
    $p\in\partial(\phi(\man{M}))$.  That is,
    \[
      \ABound=\{[p]:p\in\partial(\phi(\man{M}))\text{, where $\phi$ is an envelopment of $\man{M}$}\}.
    \]
  \end{Def}

  \begin{Def}[see \cite{Ashley2002a,Whale2010}]
    Two boundary sets $B\subset\partial(\phi(\man{M}))$ and $B'\subset\partial(\psi(\man{M}))$ are said to be in contact if for all
    open neighbourhoods $U$ of $B$ and $V$ of $B'$ we have that
    \[
      \phi^{-1}\left(U\cap\phi(\man{M})\right)\cap\psi^{-1}\left(V\cap\psi(\man{M})\right)\neq\EmptySet.
    \] 
  \end{Def}
  
  \begin{Lem}[see \cite{Ashley2002a,Whale2010}]
    Let $B\subset\partial(\phi(\man{M}))$ and $B'\subset\partial(\psi(\man{M}))$ be boundary sets. Then
    $B$ and $B'$ are in contact if and only if there exists a sequence $\{x_i\}$ in $\man{M}$ so that
    $\{\phi(x_i)\}$ has a limit point in $B$ and $\{\psi(x_i)\}$ has a limit point in $B'$.
  \end{Lem}

  \begin{Def}[see \cite{Ashley2002a,Whale2010}]
    Let $\phi:\man{M}\to\man{M}_\phi$ be an envelopment, then the
    set $\sigma_\phi=\{[p]\in\ABound:p\in\partial(\phi(\man{M}))\}$,
    is the partial cross section of $\phi$.
  \end{Def}

  We remind the reader that if $\seq s$ is a sequence in \man{M}, we mean that
  $\seq s\subset\man{M}$ and that $\seq s$ is countable.  We will not
  worry about a specific ordering of $\seq s$\footnote{The reason for this will be made clear in definition \ref{DisEm:Def.CauchyCompletion}.}.
  
  By $\seq s\to A$ we mean that there exists $x\in A$, a not 
  necessarily
  unique, accumulation point of $\seq s$. Where $A=\{x\}$ we shall write $\seq s\to x$. The reason for this non-standard notation is that,
  as points and sets are treated equivalently with respect to the Abstract Boundary, we are often interested in showing that a sequence has at least 
  one limit point
  in a particular set but not that the sequence only has limit points in that set. 
  
  We will sometimes say that a sequence has a limit point. By this we mean only that a limit point exists. 
  Where we need to mention unique limit points we say that the sequence $\seq s$ converges to $x$ or that
  $\seq s$ has the unique limit point $x$.

\section{An equivalence on the set of envelopments}\label{sec.envelopments}

    We wish to define an equivalence relation that tells us when
    two envelopments are equivalent from the point of view of the
    Abstract Boundary: that is, when they produce the same Abstract Boundary points.
    There is a natural way to express this equivalence.

  \begin{Def}\label{def.equivalentenvelopments}
    Let $\man{M}$ be a manifold, let $\phi:\man{M}\to\man{M}_\phi$
    and $\psi:\man{M}\to\man{M}_\psi$ be two envelopments of
    $\man{M}$. The envelopments $\phi$ and $\psi$ are equivalent,
    $\phi\simeq\psi$, if and only if $\sigma_\phi=\sigma_\psi$.
    That is, $\phi\simeq\psi$ if and only if for all $x\in\partial(\phi(\man{M}))$
    there exists $y\in\partial(\psi(\man{M}))$ so
    that $[y]=[x]$ and likewise for all $y\in\partial(\psi(\man{M}))$ there exists $x\in\partial(\phi(\man{M}))$ 
    so that $[x]=[y]$.
  \end{Def}

  \begin{Prop}
    The equivalence relation $\simeq$ is well defined on the set of all envelopments.
  \end{Prop}
  \begin{proof}
    This follows from the fact that $=$ is a well defined
    equivalence relation on \ABound.
  \end{proof}
  
    Looking ahead, however, we shall be working with distances on \man{M} and will need a different, yet equivalent, 
    definition that is easier to use in this setting.  With this in mind we provide the following definition and result.
    
    \begin{Def}
    Let $\Seqz=\{\seq{s}: \seq{s}$ is a sequence in $\man{M}$ with no limit points in
    $\man{M}\}$.  Where the context is clear, we will drop the
    $\man{M}$ and simply write $\Sigma_0$.
    Let $\phi:\man{M}\to\man{N}$
    be an envelopment and
    $A\subset\partial(\phi(\man{M}))$.
    Define $\Seqs{\phi}{A}$ to be the set $\{\seq{s}\in\Seqz:$
    $\phi(\seq{s})\to
    A\}$.  We will often be interested in the case when
    $A=\{a\}$, where $a\in\partial(\phi(\man{M}))$, and will write $\Seqs{\phi}{a}$ rather
    than $\Seqs{\phi}{\{a\}}$.
  \end{Def}
  
    The following lemma will allow us to give definition \ref{def.equivalentenvelopments}
    in terms of sequences
    in \man{M}.

  \begin{Lem}\label{DisEm:Lem.SequenceDefinitionOfEmbeddingContainment}
    Let $\man{M}$ be a manifold, let $\phi:\man{M}\to\man{M}_\phi$
    and $\psi:\man{M}\to\man{M}_\psi$ be two envelopments of
    $\man{M}$. Then $\sigma_{\phi}\subset \sigma_{\psi}$
    if and only if for all
    $x\in\bound[\phi]{M}$ there exists $y\in\bound[\psi]{M}$ so
    that $\Seqs{\phi}{x}=\Seqs{\psi}{y}$\footnote{Note that this lemma indicates that if $\sigma_{\phi}\subset \sigma_{\psi}$ then
                the Cauchy structure on $\man{M}$ given by $\psi$ is both larger than and compatible with the
                Cauchy structure on $\man{M}$ given by $\phi$.}.
  \end{Lem}
  \begin{proof}
    Suppose that $\sigma_\phi\subset\sigma_\psi$. Let
    $x\in\bound[\phi]{M}$, then $[x]\in\sigma_\phi$, so there exists
    $y\in\bound[\psi]{M}$ so that $[x]=[y]$; that is, $x\equiv y$.
    Now, let $\seq s\in \Seqs{\phi}{x}$, then $\phi(\seq s)\to x$, but $y\covers
    x$, so $\psi(\seq s)\to y$ and $\seq s\in \Seqs{\psi}{y}$. Therefore
    $\Seqs{\phi}{x}\subset \Seqs{\psi}{y}$. Since $x\covers y$ we can use
    the same argument to show that
    $\Seqs{\psi}{y}\subset \Seqs{\phi}{x}$ and hence $\Seqs{\phi}{x}=\Seqs{\psi}{y}$
    as required.

    Now, suppose that for all
    $x\in\bound[\phi]{M}$ there exists $y\in\bound[\psi]{M}$ so
    that $\Seqs{\phi}{x}=\Seqs{\psi}{y}$. Let $[x]\in\sigma_\phi$, where $x\in\partial(\phi(\man{M}))$, and let
    $y\in\bound[\psi]{M}$ be such that $\Seqs{\phi}{x}=\Seqs{\psi}{y}$. Let
    $\seq q\subset\man{M}$ be a sequence so that $\phi(\seq q)\to x$, then there
    exists $\seq s\in \Seqz$ so that
    $\seq s\subset\seq q$ and $\phi(\seq s)\to x$. Then, by
    construction, $\psi(\seq s)\to y$, so that $\psi(\seq q)\to y$,
    and hence $y\covers x$. Since
    $\Seqs{\phi}{x}=\Seqs{\psi}{y}$ we can use the same argument
    to show that
    $x\covers y$ and therefore
    $[x]=[y]$.  Thus $\sigma_\phi\subset
    \sigma_\psi$ as required.
  \end{proof}
  
    Next we establish a collection of equivalent definitions of
    $\simeq$.
  
  \begin{Prop}\label{DisEm:Prop.ConditionsForEquivalenceOfEmbeddings}
    Let $\man{M}$ be a manifold, let $\phi:\man{M}\to\man{M}_\phi$
    and $\psi:\man{M}\to\man{M}_\psi$ be two envelopments of
    $\man{M}$.
    The following are equivalent:
    \begin{enumerate}
      \item The envelopments $\phi$ and $\psi$ are equivalent, $\phi\simeq\psi$.
      \item There exists a homeomorphism,
        $f:\overline{\phi(\man{M})}\to
        \overline{\psi(\man{M})}$, so that $f\phi=\psi$.
      \item For all
        $x\in\bound[\phi]{M}$ there exists $y\in\bound[\psi]{M}$ so
        that $\Seqs{\phi}{x}=\Seqs{\psi}{y}$ and for all
        $p\in\bound[\psi]{M}$ there exists $q\in\bound[\phi]{M}$ so
        that $\Seqs{\psi}{p}=\Seqs{\phi}{q}$.
    \end{enumerate}
  \end{Prop}
  
\begin{proof}
    \begin{description}
      \item[$\mathbf{3\Leftrightarrow 1}$] Apply lemma \ref{DisEm:Lem.SequenceDefinitionOfEmbeddingContainment}
        twice.

      \item[$\mathbf{1 \Rightarrow 2}$] Since $\phi\simeq\psi$ we know that
        $\sigma_\phi=\sigma_\psi$.
        Therefore for each $x\in\bound[\phi]{M}$ there exists
        $y\in\bound[\psi]{M}$ so that $x \equiv y$; denote this $y$ by $x_\psi$. Define
        $f:\overline{\phi\man{M}}\to\overline{\psi\man{M}}$ by
        \[
          f(x)=\left\{
          \begin{aligned}
            &\psi\phi^{-1}(x) &&\text{if}\ x\in \phi\man{M}\\
            &x_\psi &&\text{otherwise.}
          \end{aligned}
          \right.
        \]
        Since $\psi\phi^{-1}$ is a homeomorphism,
        we need only consider $f|_{\bound[\phi]{M}}$.

        We need to show that $f$ is well defined.
        Suppose that
        there exist $u,v\in\bound[\psi]{M}$ so that $x\equiv u$ and
        $x\equiv v$, where $x\in\bound[\phi]{M}$. Since, $\equiv$ is an equivalence relation
        $u\equiv
        v$. By theorem \ref{Thm:Abstract Boundary Covering Relation -
      sequences} and as $\man{M}_\psi$ is hausdorff, we can conclude that $u=v$. Therefore
        $f$ is well defined.

        We need to show that $f$ is surjective. Let
        $y\in\bound[\psi]{M}$ then, since $\phi\simeq\psi$, there
        exists $x\in\bound[\phi]{M}$ so that $[y]=[x]$ and hence
        $f(x)=y$. Therefore $f$ is surjective.

        We need to show that $f$ is injective. Suppose
        that $x,y\in\bound[\phi]{M}$ are such that $f(x)=f(y)$, then $x\equiv
        f(x)=f(y)\equiv y$ so that $x\equiv y$. As before by theorem \ref{Thm:Abstract Boundary Covering Relation -
      sequences} and as
        $\man{M}_\phi$ is
        hausdorff, we know that $x=y$.

        We need to show that $f$ is continuous. Since $\man{M}$ is first countable
        we can do this by showing that $f$ is sequentially continuous. 
        The proof that $f$ is sequentially continuous is long. It is divided into five sections.
        In the first section we show that $f$  is continuous on $\phi\man{M}$. The second
        section shows that any for any sequence $\{x_i\}\subset\phi\man{M}$ with unique limit point
        $x\in \partial(\phi(\man{M}))$ we have that $\{f(x_i)\}$ converges to $f(x)$. The arguments of the
        section also shows that a similar statement holds for $f^{-1}$. The third section shows that
        for any sequence $\{x_i\}\subset\partial(\phi(\man{M}))$ converging to $x$, necessarily
        in $\partial(\phi(\man{M}))$, it is the case that $f(x)$ is a limit point of
        $\{f(x_i)\}$. The fourth section shows that the sequence $\{f(x_i)\}$ of the third section
        uniquely converges to $f(x)$. The fifth section considers sequences in $\overline{\phi\man{M}}$
        whose elements are not restricted to either $\phi\man{M}$ or $\partial\phi\man{M}$. These
        arguments demonstrate that $f$ is sequentially continuous and therefore continuous.
        Since we shall repeat the arguments of earlier paragraphs in later sections we will number all
        paragraphs to make reference to the arguments easier.

        \textbf{1} First, since $f$ restricted to $\phi\man{M}$ is
        $\psi\phi^{-1}$ we only need consider sequences that converge to
        points in $\bound[\phi]{\man{M}}$.

        \textbf{2} Second, suppose that $\{x_i\}\subset\phi\man{M}$
        converges to $x\in\bound\phi\man{M}$.  Since
        $f(x)\equiv x$ we know that $\{f(x_i)\}$ must have $f(x)$
        as a limit point (by theorem \ref{Thm:Abstract Boundary Covering Relation -
        sequences}). Any subsequence $\{p_i\}$ of $\{x_i\}$ must
        also be such that $\{f(p_i)\}$ has $f(x)$ as a limit
        point by the same reasoning, since $\{p_i\}$ must converge uniquely to $x$.
        We will show
        that $f(x)$ is unique. Suppose that there exists $q\in\overline{\psi\man{M}}$ and a
        subsequence $\{q_i\}$ of $\{x_i\}$ so that $\{f(q_i)\}$
        converges to $q$.  Since $\{q_i\}\subset\{x_i\}$ we know
        that $\{q_i\}$ converges to $x$ and as $f(x)\equiv x$ we
        know that $f(x)$ is an accumulation point of
        $\{f(q_i)\}$ but, by construction, $\{f(q_i)\}$ converges to $q$.
        Therefore $q=f(x)$. Thus for all sequences $\seq s$ lying in $\phi\man{M}$ so that 
        $\seq s\to x\in\partial\phi\man{M}$ uniquely, we know that $f(\seq s)\to f(x)$ uniquely.  Since the argument of this
        paragraph can also be applied to $f^{-1}$ we know that for all sequences $\seq s\subset\psi\man{M}$ so that
        $\seq s\to y\in\partial\psi\man{M}$ uniquely, we have that $f^{-1}(\seq s)\to f^{-1}(y)$ uniquely.  We use these facts below.

        \textbf{3} Third, suppose that $\{x_i\}$ is a sequence in
        $\bound[\phi]{\man{M}}$ converging to $x$ and suppose that
        $\{f(x_i)\}$ has no limit points, then we can choose an
        open neighbourhood $V$ of $f(x)$ so that for
        all $i$, $f(x_i)\not\in V$.  For each $i$ choose a
        sequence $\{f(y_j^i)\}\subset \psi\man{M}$ that converges uniquely
        to $f(x_i)$.  As $\{f(y_j^i)\}$ converges to $f(x_i)$ we
        know that $\{f(y_j^i)\}\cap V$ must be finite, hence without
        loss of generality we may assume that for all $i,j$, $\{f(y_j^i)\}\cap
        V=\EmptySet$ and therefore that for all $i,j$,
        $f(y_j^i)\not\in V$. Using the same techniques as in paragraph \textbf{2} we can show
        that each $\{y_j^i\}$ must converge uniquely to $x_i$.  Form a
        new sequence, $\{s_k\}=\bigcup_{i,j}\{y_j^i\}$.  By
        construction, since each $\{y_j^i\}$ converges to $x_i$, we know
        that $\{s_k\}$ has $x$ as a limit point and as
        $\{s_k\}\subset\phi\man{M}$ we know that $\{f(s_k)\}$ has
        $f(x)$ as a limit point.  This implies that there exists
        an infinite subsequence, $\{f(q_l)\}$ of $f(s_k)$ so that
        for all $l$, $f(q_l)\in V$.  This is a contradiction,
        however, as for each $l$ there exists $i_l$ and $j_l$ so
        that $q_l=y^{i_l}_{j_l}$ where we know that $f(q_l)\in V$ and
        that $f(y^{i_l}_{j_l})\not\in V$. Therefore $\{f(x_i)\}$ must
        have $f(x)$ as a limit point.

        \textbf{4} Fourth, suppose that $\{f(x_i)\}$ also has a limit point
        $q\in\bound[\phi]{\man{M}}$, where $\{x_i\}$ is the sequence of paragraph \textbf{3}.
        We may choose a sequence of
        open
        neighbourhoods, $V_i$, so that $x_i\in V_i$, for all
        $j\not = i$, $x_j\not\in V_i$ and for all $i, j,\ i\not = j$,
        $\overline{V_i}\cap\overline{V_j}=\EmptySet$.  Let
        $\{y_j^i\}\subset\phi\man{M}$ be a sequence that
        converges uniquely to $x_i$ and is such that for all $j$,
        $y^i_j\in V_i$. Let
        $\{s_k\}=\bigcup_{i,j}\{y_j^i\}$ be a new sequence formed
        from the union of the $y_j^i$'s.  Since $\{f(x_i)\}$ has
        $q$ as a limit point we know that $\{f(s_k)\}$ must also
        have $q$ as a limit point.  From paragraph \textbf{2} 
        we know that $f^{-1}(q)$ must be a limit point of
        $\{s_k\}$. By construction this implies that $f^{-1}(q)$
        is either equal to $x_i$ for some $i$ or equal to $x$.

        \textbf{5} If $f^{-1}(q)=x$ then we are done, so suppose that there
        exists $l$ so that $q=f(x_l)$.  Since $q$ is a limit point of $\{f(x_i)\}$
        we can choose a subsequence, $\{q_r\}$ of
        $\{f(s_k)\}$ so that $q_r\in\{f(y_j^r)\}$ and $\{q_r\}$ uniquely converges to
        $q$. From paragraph \textbf{2} we know that $\{f^{-1}(q_r)\}$ must have
        $f^{-1}(q)=x_l$ as a unique limit point. This implies that
        $\{f^{-1}(q_r)\}\cap V_l$ must be infinite. But by
        construction we know that for all $r\not =l$,
        $y^r_j\not\in V_l$, and since $q_r=y^r_{j}$ for some $j$
        we know that $\{f^{-1}(q_r)\}\cap V_l$ is either empty or contains
        only the element $f^{-1}(q_l)$.  Therefore we have a contradiction
        and $q=f(x)$ as required.

        \textbf{6} Fifth, suppose that $\{x_i\}$ lies in
        $\overline{\phi\man{M}}$ and that it uniquely converges to
        $x\in\bound[\phi]{M}$.  If either
        $\{x_i\}\cap\bound[\phi]{\man{M}}$ or
        $\{x_i\}\cap\phi\man{M}$ is finite we can use the
        arguments above to show that $\{f(x_i)\}$ uniquely
        converges to $f(x)$, so suppose that neither set is
        finite. In this case we know that both sequences $f(\{x_i\}\cap\bound[\phi]{\man{M}})$
        and
        $f(\{x_i\}\cap\phi\man{M})$ must uniquely converge to
        $f(x)$,  and therefore $\{f(x_i)\}$ must also uniquely
        converge to $f(x)$. Hence $f$ is continuous.

        Since $x\covers f(x)$, $\phi\psi^{-1}$ is continuous and $f^{-1}$
        is bijective the
        same argument can be applied to show that $f^{-1}$ is
        continuous and therefore $f$ is a homeomorphism.
      \item[$\mathbf{2\Rightarrow 3}$] Suppose that there exists a
      homeomorphism,
      $f:\overline{\phi\man{M}}\to\overline{\psi\man{M}}$ so that
      $f\phi=\psi$. Let $x\in\bound[\phi]{\man{M}}$ and $y=f(x)$.
      Since $f\phi=\psi$, it is clear that $y\in\bound[\psi]{\man
      M}$. Let $\seq s\in\Seqs{\phi}{x}$. By the continuity of $f,
      \seq s\in\Seqs{\psi}{y}$, thus $\Seqs{\phi}{x}\subset
      \Seqs{\psi}{y}$.

      Now, let $\seq t\in\Seqs{\psi}{y}$. By the continuity of
      $f^{-1}, \seq t\in\Seqs{\phi}{x}$. Hence
      $\Seqs{\psi}{y}\subset\Seqs{\phi}{x}$.

      It follows that $\Seqs{\phi}{x}=\Seqs{\psi}{y}$. By a
      similar argument, it can be shown that for all
      $p\in\bound[\psi]{\man M}$ there exists
      $q\in\bound[\phi]{\man{M}}$ so that
      $\Seqs{\psi}{p}=\Seqs{\phi}{q}$.
    \end{description}
  \end{proof}
  
    The technique employed in $\mathbf{1 \Rightarrow 2}$ can be very
    useful when working with the Abstract Boundary.

  \begin{Cor}
    Let $\phi:\man{M}\to\man{M}_\phi$ and
    $\psi:\man{M}\to\man{M}_\psi$ be envelopments, let
    $x\in\bound[\phi]{\man{M}}, y\in\bound[\psi]{\man{M}}$ be such
    that $x\equiv y$ and let $U\subset\man{M}$. Then
    $x\in\overline{\phi(U)}$ if and only if
    $y\in\overline{\psi(U)}$.
  \end{Cor}
  \begin{proof}
    Suppose that $x\in\overline{\phi(U)}$. Then there exists a
    sequence $\{p_i\}\subset U$ so that $\{\phi(p_i)\}\to x$.
    From theorem \ref{Thm:Abstract Boundary Covering Relation -
      sequences} we know that $\{\psi(p_i)\}$ must have $y$ as a
      limit point. Since $\{\psi(p_i)\}\subset\psi(U)$ then
      $y\in\overline{\psi(U)}$.  The same argument can be applied
      in the reverse direction.
  \end{proof}

  \subsection{Cauchy structures}\label{envelopments cauchy}
        Note that the proof of proposition \ref{DisEm:Prop.ConditionsForEquivalenceOfEmbeddings} is very similar to the proof that
        uniformly continuous functions $f:X\to Y$ into a complete space $Y$ extend uniquely to the completion of $X$. This similarity is
        not accidental, even though in our case the functions involved are not necessarily uniformly continuous. The similarity exists because there is a more
        general extension theorem (see \cite{Reed1971}) regarding Cauchy continuous functions which does apply in our case.

        The more general theorem is best expressed in the
        language of Cauchy spaces \cite{LowenColebunders1989}.
        Since the majority of our intended audience are unlikely to be familiar with Cauchy spaces and as explicit proofs are instructive, we have given full proofs of the main
        results, propositions \ref{DisEm:Prop.ConditionsForEquivalenceOfEmbeddings} and \ref{DisEm:Prop.DistanceEquivalenceStatmentsNotInDis}, without
        appealing to Cauchy spaces. For those readers who are familiar we include below a brief review of the needed material, the relevant extension
        theorem and the alternative proof of proposition \ref{DisEm:Prop.ConditionsForEquivalenceOfEmbeddings}. At the end of section \ref{sec.distances}
        we give an alternative proof of proposition \ref{DisEm:Prop.DistanceEquivalenceStatmentsNotInDis}.

        Cauchy structures on a set are usually defined in terms of filters.
        In any first countable topological space, however, filters are equivalent to sequences. We provide the following definition:
        \begin{Def}
                A Cauchy structure on a metric space $T$ is a subset $S$
                of the set $\Sigma(T)$ of all sequences in $T$ so that the set of filters corresponding to the sequences in $S$ generates
                a Cauchy structure
                in the sense of \cite{LowenColebunders1989}. That is, the set $\mathcal{F}$ of filters corresponding to the sequences in $S$ is 
                such that
                \begin{enumerate}
                        \item for each $x\in T$ the ultrafilter at $x$ is in $\mathcal{F}$,
                        \item if $F\in\mathcal{F}$ and $F\subset G$, where $G$ is a filter, then
                                                $G\in\mathcal{F}$,
                        \item if $F,G\in\mathcal{F}$ and every element of $F$ intersects every element of $G$ then $F\cap G\in\mathcal{F}$.
                \end{enumerate}
                We say that a Cauchy structure $S$ on a metric space $T$ is compatible with the topological structure if, for every $\seq s\in S$,
                we have that either $\seq s$ has no limit points or every subsequence $\seq q$ of $\seq s$ has a unique limit point with respect to the topology. 
                Whenever we impose a Cauchy structure on a metric space we
                will implicitly assume that the topology and Cauchy structure are compatible.
        \end{Def}
        We have two specific situations to consider. Let $d$ be a distance on a manifold $\man{M}$, then the set of all sequences that are
        Cauchy with respect $d$ is a Cauchy structure. Let $\phi:\man{M}\to\man{M}_\phi$ be an envelopment of a manifold $\man{M}$, then the set of
        all sequences $\seq s$ in $\man{M}$ so that every subsequence $\seq q\subset\seq s$ converges under $\phi$
        is a Cauchy structure.
        \begin{Def}
                A function $f:X\to Y$ between metric spaces is Cauchy continuous with respect to the Cauchy structures $C_X$ and $C_Y$
                on $X$ and $Y$ respectively if, for all $\seq s\in C_X$, $f(\seq s)\in C_Y$.
        \end{Def}
        \begin{Def}
                Let $Y$ be a metric space with Cauchy structure $S$. Then $Y$ is said to be complete with respect to $S$
                if and only if every sequence $\seq s\in S$ converges to a 
                point with respect to the topology on $Y$. Where the selection of a Cauchy structure $S$ is clear, we shall simply say that $Y$ is complete.
        \end{Def}
        \begin{Def}
                Let $X$ and $Y$ be metric spaces. Then $Y$ is a completion of $X$ if $Y$ is complete and there exists $f:X\to Y$, a Cauchy continuous
                map, so that $f(X)$ is dense in $Y$. 
        \end{Def}
        In particular, if $\phi:\man{M}\to\man{M}_\phi$ is an envelopment then $\man{M}$ carries the Cauchy structure 
        derived from $\phi$ discussed above. With respect
        to this structure, the topological space $\overline{\phi(\man{M})}$ is a completion of $\man{M}$. Likewise,
        the usual Cauchy completion of a manifold is a completion, in the sense above, with respect to the Cauchy structure induced by a
        distance and the inclusion embedding.
        
        We have the following result:
        \begin{Thm}\label{newTHerem}
                Let $Y$ be a metric space which is a completion of the metric space $X$, with respect to the
                Cauchy continuous map $k:X\to Y$ and the Cauchy structures $C_X$ and $C_Y$ on $X$ and $Y$ respectively. Let
                $Z$ be a metric space complete with respect to the Cauchy structure $C_Z$ on $Z$.
                Then a Cauchy continuous function $f:X\to Z$ has a unique continuous and Cauchy continuous extension $\hat{f}:Y\to Z$ so that
                $\hat{f}\circ k= f$.
        \end{Thm}
        \begin{proof}
                Refer to \cite{Reed1971} or \cite{LowenColebunders1989}.
        \end{proof}

        We are now able to give our alternate proof of proposition \ref{DisEm:Prop.ConditionsForEquivalenceOfEmbeddings}.
        
        \begin{proof}[Alternative proof of proposition \ref{DisEm:Prop.ConditionsForEquivalenceOfEmbeddings}]
    \begin{description}
      \item[$\mathbf{3\Leftrightarrow 1}$] Apply lemma \ref{DisEm:Lem.SequenceDefinitionOfEmbeddingContainment}
        twice.

      \item[$\mathbf{1 \Rightarrow 2}$]
        Define
        $f:\man{M}\to\overline{\psi(\man{M})}$ by
        $
          f(x)=\psi(x)
        $ for all $x\in \man{M}$. Equip $\man{M}$ and $\overline{\psi(\man{M})}$ with the Cauchy structures $C_\phi$ and $C_{\psi}$
        induced by inclusion into $\man{M}_\phi$ and $\man{M}_\psi$, respectively.
        That is, a sequence $\seq s$ is Cauchy in $\man{M}$ if and only if $\phi(\seq s)$ converges in $\overline{\phi(\man{M})}$ and
        a sequence $\seq s$ is Cauchy in $\overline{\psi(\man{M})}$ if and only if $\seq s$ converges in $\overline{\psi(\man{M})}$.
        
        We will now show that $f$ is Cauchy continuous.
        Let $\seq s\in C_\phi$. If $\phi(\seq s)$ converges in $\phi(\man{M})$ then
        $\psi(\seq s)$ will also converge in $\psi(\man{M})$, by the continuity of $\psi\circ\phi^{-1}$.
        So suppose then that $\phi(\seq s)\to x\in\partial(\phi(\man{M}))$ uniquely. By lemma      
        \ref{DisEm:Lem.SequenceDefinitionOfEmbeddingContainment} there exists $y\in\partial(\psi(\man{M}))$ so that
        $\Seqs{\phi}{x}=\Seqs{\psi}{y}$. Thus, as $\seq s\in\Seqs{\phi}{x}$, we know that
        $\seq s\in \Seqs{\psi}{y}$. Hence $\psi(\seq s)$ has $y$ as a limit point.
        To show that $f(\seq s)$ is Cauchy, however, we must show that the limit point of
        $\psi(\seq s)$ is unique.
        
        Let us first suppose that $\psi(\seq s)$ has a subsequence $\psi(\seq p)$ with no limit points. Since $\seq p\subset \seq s$
        we know that $\seq p\in \Seqs{\phi}{x}$ and therefore that $\seq p\in \Seqs{\psi}{y}$. This contradicts our assumption that
        $\psi(\seq p)$ has no limit points. Likewise suppose that $\psi(\seq s)$ also has $p\in\overline{\psi(\man{M})}$, $p\neq y$, as a limit point.
        Then there must exist some subsequence $\psi(\seq p)$ of $\psi(\seq s)$ so that $\psi(\seq p)\to p$ uniquely. Since $\seq p\subset\seq s$,
        we know that
        $\seq p\in \Seqs{\phi}{x}$ and hence $\seq p\in\Seqs{\psi}{y}$. Once again this is a contradiction. It follows that
        $\psi(\seq s)$ converges to $y$ in $\partial(\psi(\man{M}))$ and so $\psi(\seq s)\in C_\psi$.
        Therefore the function $f$ is Cauchy continuous.
        
        Thus, by theorem \ref{newTHerem}, $f$ has a unique extension $\hat{f}:\overline{\phi(\man{M})}\to\overline{\psi(\man{M})}$
        so that $\hat{f}\phi=\psi$. By the
        uniqueness of $\hat{f}$ and as the argument is symmetric in $\phi$ and $\psi$, it must be the case that
        $\hat{f}$ is a homeomorphism.
        
      \item[$\mathbf{2\Rightarrow 3}$] Suppose that there exists a
      homeomorphism,
      $f:\overline{\phi(\man{M})}\to\overline{\psi(\man{M})}$ so that
      $f\phi=\psi$. Let $x\in\bound[\phi]{\man{M}}$ and $y=f(x)$.
      Since $f\phi=\psi$, it is clear that $y\in\bound[\psi]{\man
      M}$. Let $\seq s\in\Seqs{\phi}{x}$. By the continuity of $f,
      \seq s\in\Seqs{\psi}{y}$, thus $\Seqs{\phi}{x}\subset
      \Seqs{\psi}{y}$.
      Now, let $\seq t\in\Seqs{\psi}{y}$. By the continuity of
      $f^{-1}, \seq t\in\Seqs{\phi}{x}$. Hence
      $\Seqs{\psi}{y}\subset\Seqs{\phi}{x}$.

      It follows that $\Seqs{\phi}{x}=\Seqs{\psi}{y}$. By a
      similar argument, it can be shown that for all
      $p\in\bound[\psi]{\man M}$ there exists
      $q\in\bound[\phi]{\man{M}}$ so that
      $\Seqs{\psi}{p}=\Seqs{\phi}{q}$.
    \end{description}
  \end{proof}

\section{An equivalence on a class of distances}\label{sec.distances}

    We now have a way to describe when two envelopments produce the same
    Abstract Boundary points in terms of sequences in $\man{M}$.  We present here
    an equivalence relation on a set of distances on $\man{M}$.  This relation mirrors 
    the one in the previous section and will allow us, in the next section, to define a one-to-one correspondence
    between the two sets of equivalence classes.
  
    First we need to define the structures with which we will be working.

  \begin{Def}
    Let $d:\man{M}\times\man{M}\to\mathbb{R}$ be a distance. Let
    $\Cau=\{\seq s\in \Seqz:\seq s $ is Cauchy with respect to
    $d\}$\footnote{The set $C(d)$ captures the portion of the Cauchy structure given by $d$ which is relevant for the $a$-boundary.}
  \end{Def}

  \begin{Def}\label{DisEm:Def.CauchyCompletion}
    Let $d:\man{M}\times\man{M}\to\mathbb{R}$ be a distance. Let
    $\man{M}^d$ denote the Cauchy completion of $\man{M}$ with respect
    to $d$.
    We think of $\man{M}^d$ as the set of all Cauchy sequences in
    $\man{M}$ under an equivalence relation.  Two Cauchy sequences
    $\seq u$ and $\seq v$ in $\man{M}$ are equivalent if
    $\lim_{i,j\to\infty}d(u_i,v_j)=0$\footnote{As $\seq u$ and $\seq v$ are Cauchy sequences the same limit will
    be achieved regardless of the choice of ordering of each sequence.}.   Let $\seq u$ be a Cauchy sequence in $\man{M}$, then we shall
    denote the equivalence class of $\seq u$ by $[\seq u]_d$. Where unambiguous
    we will drop the subscript and simply write $[\seq u]$.
    We topologise $\man M^d$ by extending the distance $d$ to $\man M^d$.
    Let $d^*:\man M^d\times \man M^d\to\mathbb{R}$ be such that
    $d^*([\seq u],[\seq v])=\lim_{i,j\to\infty}d(u_i,v_j)$. Then there exists an isometry
    $\imath_d:\man{M}\to {\man M}^d$ so that $\imath_d(x)=[\seq w_x]$ where $\seq w_x$
    is the constant sequence at $x$.
    Note that, in general, $\man M^d$ is not necessarily a manifold with boundary.
  \end{Def}

    For utility, we give the following trivial result.

  \begin{Lem}\label{DisEm:Lem.SequencesInCauchyCompletetion}
    Let $\seq u=\{u_i\}$  be a Cauchy sequence in \man{M} with respect to $d$,
    then $\{\imath_d(u_i)\}$ $\to [\seq v]$
    if and only if $\seq u\in[\seq v]$.
  \end{Lem}
  \begin{proof}
    Suppose that $\{\imath_d(u_i)\}_i\to [\seq v]$.
    Then for all $\epsilon>0$ there exists
    $i^*\in\mathbb{N}$ so that for all $i>i^*$,
    $d^{*}(\imath_d(u_i),[\seq v])<\epsilon$. The condition $d^{*}(\imath_d(u_i),[\seq
    v])<\epsilon$ is equivalent to
    $\lim_{i,j\to\infty}d(u_i,v_j)<\epsilon$, however, and so
    $\lim_{i,j\to\infty}d(u_i,v_j)=0$. Thus $\seq u\in[\seq v]$.

    Suppose that $\seq u\in[\seq v]$. Then
    $\lim_{i,j\to\infty}d(u_i,v_j)=0$, so that for all
    $\epsilon>0$ there exists $i^*\in\mathbb{N}$ so that if
    $i>i^*$, then $\lim_{i,j\to\infty}d(u_i,v_j)<\epsilon$. As before this
    implies that $d^*(\imath_d(u_i),[\seq v])<\epsilon$, from which we can see
    that $\{\imath_d(u_i)\}_i\to[\seq v]$.
  \end{proof}

    Now we define an equivalence relation on the set of all
    distances on \man{M}.

  \begin{Def}
    Let $d$ and $d'$ be distances on \man{M}. Then $d$ and $d'$
    are equivalent, $d\simeq d'$, if and only if
    $\Cau[d]=\Cau[d']$\footnote{Note that this is slightly different from Cauchy equivalence since we are only concerned with Cauchy sequences that have no limit points in $\man{M}$.}.
  \end{Def}

    At some point we must link distances with envelopments, as that is the aim of 
    this paper. Unfortunately, it is reasonably easy to give examples of distances
    on a manifold that have no relation to any envelopments%
    \footnote{For example, let $\man{M}=\mathbb{R}^+$ with the embedding given by
          $\phi(x)=(x,\sin(\frac{1}{x}))$ into $\mathbb{R}^2$. Then the usual euclidean distance on $\mathbb{R}^2$ induces the distance 
          $d(x,y)=\sqrt{(x-y)^2+(\sin\frac{1}{x}-\sin\frac{1}{y})^2}$ on \man{M}. We note that the sequences $\{x_i=\frac{1}{2i\pi}\}$, 
          $\{y_i=\frac{2}{4i\pi+\pi}\}$ and 
          $\{z_i=\frac{2}{4i\pi+3\pi}\}$
          are all Cauchy with respect to $d$. Thus for $d$ to be induced by an envelopment $\phi$ it must be the
          case that the sequences $\{\phi(x_i)\}$, $\{\phi(y_i)\}$ and $\{\phi(z_i)\}$ each have a unique limit point.
          Moreover any envelopment of $\man{M}$ must be into either $\mathbb{R}$ or $S^1$. It will, therefore, have either zero, one or two boundary points.
          This implies that at least one of the pairs
          of sequences 
          $(\{\phi(x_i)\},\{\phi(y_i)\})$, $(\{\phi(x_i)\},\{\phi(z_i)\})$ and $(\{\phi(y_i)\},\{\phi(z_i)\})$ must share the same limit point. We 
          can calculate, however,
          that $\lim_{i\to\infty}d(x_i,y_i)=1$, $\lim_{i\to\infty}d(x_i,z_i)=1$ and $\lim_{i\to\infty}d(y_i,z_i)=2$. 
          This implies that none of the sequences $\{\phi(x_i)\}, \{\phi(y_i)\}, \{\phi(z_i)\}$ share a limit point. As this is a contradiction,
          $d$ cannot be induced by any envelopment $\phi$. Hence we must restrict the set of distances in which we are interested.}.

  \begin{Def}\label{DisEmb:Def.EnvelopebleDistance}
    Let $d:\man{M}\times\man{M}\to\mathbb{R}$ be a distance on
    $\man{M}$ and define $E(d)$ to be the set of all non-trivial envelopments
    $\phi:\man{M}\to\man{M}_\phi$ so that there exists a complete distance
    $d':\man{M}_\phi\times\man{M}_\phi\to\mathbb{R}$ with the property
    that $d'|_{\phi(\man{M})\times\phi(\man{M})}\simeq d$.
  \end{Def}

  \begin{Def}
    Let $\Dis=\{d: E(d)\neq\EmptySet\}$.
  \end{Def}

    It is interesting that $\Dis=\EmptySet$ is equivalent to the
    compactness of $\man{M}$.

  \begin{Prop}
    Let $\man{M}$ be a manifold.  The set $\Dis$ is empty if and
    only if $\man{M}$ is compact.
  \end{Prop}
  \begin{proof}
    $\Leftarrow$
    Since $\man{M}$ is compact there are no non-trivial
    envelopments of \man{M}.  Hence for all distances $d$ on
    \man{M}, $E(d)=\EmptySet$.  Therefore for all $d$, $d\not\in\Dis$
    and $\Dis$ must be empty.

    $\Rightarrow$
    Suppose that \man{M} is not compact.  Then there exists a
    sequence $\seq s$ in $\man{M}$ that has no limit points in \man{M}.  By the Endpoint 
    Theorem (see \cite{Ashley2002a}), we
    may use this sequence to construct an envelopment,
    $\phi:\man{M}\to\man{M}_\phi$ so that $\phi(\seq s)$ converges
    to some point on $\bound[\phi]{\man{M}}$.

    Choose a distance $d_\phi:\man{M}_\phi\times\man{M}_\phi\to\mathbb{R}$ so
    that $d_\phi$ is complete.  Define the distance
    $d:\man{M}\times\man{M}\to\mathbb{R}$ by
    $d(x,y)=d_\phi(\phi(x),\phi(y))$. Then $d\simeq
    d_\phi|_{\phi(\man{M})\times\phi(\man{M})}$ so that $\phi\in E(d)$
    and hence $d\in D(\man{M})$.

    Taking the contrapositive of
    this result we see that $D(\man{M})=\EmptySet$ implies that
    $\man{M}$ is compact.
  \end{proof}

    Just as for envelopments, there are several different ways to
    express the relation $\simeq$; again we give only those that
    are useful.

  \begin{Prop}\label{DisEm:Prop.DistanceEquivalenceStatmentsNotInDis}
    Let $d$ and $d'$ be distances on $\man{M}$, then $d\simeq d'$
    if and only if there exists a homeomorphism $f:\man{M}^d\to \man{M}^{d'}$ so that
        $f\imath_d=\imath_{d'}$.
  \end{Prop}
  \begin{proof}
        The proof of this result follows the same format as the proof of
        step $\mathbf{1 \Rightarrow 2}$ in proposition \ref{DisEm:Prop.ConditionsForEquivalenceOfEmbeddings}.
  
    $\Rightarrow$

    Since $d\simeq d'$, we know that $\Cau[d]=\Cau[d']$, and hence we can
    define $f:\man{M}^d\to\man{M}^{d'}$ by $f([\seq s]_d)=[\seq
    s]_{d'}$\footnote{This is the unique $f$ given by the lifting of
    $\imath_d$ to $\man{M}^{d'}$.} where $\seq s$ is Cauchy with respect to $d$. Note that
    $\seq s$ may have a limit point in $\man{M}$, and hence would not be in $\Seqz$.

    Because $\Cau[d]=\Cau[d']$ we can easily see that $f$ is surjective.  If 
    $f([\seq u])=f([\seq v])$ then the sequence $\seq w$ given by;
    \begin{equation*}
                        w_i =
                        \begin{cases}
                                u_{\frac{i}{2}}&\text{iff}\ i=2n\\ 
                                v_{\frac{i+1}{2}}&\text{iff}\ i=2n+1
                        \end{cases}     
                \end{equation*}
    is such that $\seq w\in\Cau[d']$.  Therefore $\seq w\in\Cau[d]$ and $[\seq u]_d=[\seq v]_d$.
    We can conclude that $f$ is
    bijective and by definition we know that
    $f\imath_d=\imath_{d'}$. Thus $f$ is continuous on
    $\imath_d\man{M}$.

    Since $\man{M}^d$ is first countable to show that $f$ is continuous everywhere
                it is enough to show that $f$ is sequentially continuous.
                Since $f$ is
    continuous on $\imath_d\man{M}$, we need only consider
    sequences in $\man{M}^d$ that have limit points in
    $\bound[\imath_d]{\man{M}}$.
    
    The proof that $f$ is sequentially continuous is long. We have
    divided the proof into four sections. The first section
    shows that for any sequence $\seq s\subset\imath_d\man{M}$ so that $\seq s\to x\in\partial\imath_d\man{M}$ uniquely
    we have that $f(\seq s)\to f(x)$ uniquely. The second section show that if
    $\seq s\subset\partial\imath_d\man{M}$ converges uniquely to $x\in\partial\imath_d\man{M}$
    then $f(\seq s)$ has $f(x)$ as a, not necessarily unique, limit point. The third
    section demonstrates that for $\seq s$ as in the second section the limit point $f(x)$ of
    $f(\seq s)$ is unique. The fourth section considers sequences in $\man{M}^d$ without
    restriction to either $\imath_d\man{M}$ or $\partial\imath_d\man{M}$. We shall sometimes
    repeat the arguments of an earlier section in later portions of the proof. To aid reference we have
    numbered the paragraphs.

    \textbf{1} First, let $\{x_i\}$
    be a sequence in $\imath_d\man{M}$ that converges uniquely to
    $x\in\bound[\imath_d]{M}$. Since
    $x_i\in\imath_d\man{M}$ there must exist $y_i\in\man{M}$ so
    that $\imath_d(y_i)=x_i$. By lemma
    \ref{DisEm:Lem.SequencesInCauchyCompletetion} we know that
    $x=[\{y_i\}]_d$.  Therefore
    $f(x)=f([\{y_i\}]_d)=[\{y_i\}]_{d'}$. Once again by lemma
    \ref{DisEm:Lem.SequencesInCauchyCompletetion} we know that
    $\{\imath_{d'}(y_i)\}$ converges uniquely to
    $[\{y_i\}]_{d'}=f(x)$.  But
    $f(x_i)=f(\imath_d(y_i))=\imath_{d'}(y_i)$, and therefore
    $f(x_i)\to f(x)$ as required. Thus for all sequences $\seq s$ lying in $\imath_d\man{M}$ so that
    $\seq s\to x\in\bound[\imath_d]{M}$ uniquely we know that $f(\seq s)\to f(x)$ uniquely. Since the argument
    of this paragraph can also be applied to $f^{-1}$ we know that for all sequences $\seq s\subset\imath{d'}\man{M}$
    so that $\seq s\to t\in\bound[\imath_{d'}]{M}$ uniquely, we have that $f^{-1}(\seq s)\to f^{-1}(y)$ uniquely. We use
    these facts below.

    \textbf{2} Second, let $\{x_i\}$
    be a sequence in $\bound[\imath_d]{M}$ that converges uniquely to
    $x\in\bound[\imath_d]{M}$ and suppose that $\{f(x_i)\}$ has no limit points in $\man{M}^{d'}$.
    We may choose an open neighbourhood $V$ of $f(x)$ so that for
    all $i$, $f(x_i)\not\in V$.  For each $i$ we can choose a
    sequence $\{f(\imath_d y^i_j)\}$ converging uniquely to $f(x_i)$.  This
    implies that $\{f(\imath_d y^i_j)\}\cap V$ must be finite and hence,
    without loss of generality, we can assume that $\{f(\imath_d y^i_j)\}\cap
    V=\EmptySet$ and therefore that for all $i,j$, $f(\imath_d y^i_j)\not\in
    V$.  From paragraph \textbf{1} we can conclude that for
    each $i$ the sequence $\{\imath_d y^i_j\}$ must converge uniquely to
    $x_i$.  Now, by construction, for the sequence
    $\seq s=\bigcup_{i,j}\{\imath_d(y^i_j)\}$, we can conclude that $f(\seq s)$
    does not have
    $f(x)$ as a limit point. The
    sequence $\seq s$ has $x_i$ as a limit point for each $i$ and
    as $\{x_i\}\to x$ we know that $\seq s$ has $x$ as a limit
    point. Choose a
    subsequence, $\seq p=\{\imath_d(p_k)\}$ of $\seq s$ so that
    $p_k\in\man{M}$ and $\seq p\to x$.  Hence, from above, we know
    that $x=[\{p_k\}]_d$ and that $\{f(\imath_d p_k)=\imath_{d'}(p_k)\}\to
    [\{p_k\}]_{d'}=f(x)$. This is a contradiction, since $\{f(\imath_d p_k)\}
    \to f(x)$ implies that $f(\seq p)\cap V\not =\EmptySet$, but
    $\seq p\subset\seq s$ and $f(\seq s)\cap V=\EmptySet$. Therefore $\{f(x_i)\}$
    has $f(x)$ as a
    limit point.

    \textbf{3} Third, we will now show that $f(x)$ is the unique limit point of
    $\{f(x_i)\}$, where $\{x_i\}$ and $x$ are the sequence and point of paragraph \textbf{2}.
    Suppose that $q\in\man{M}^{d'}$ is a limit point of $\{f(x_i)\}$; since $\{x_i\}\subset\partial\imath_d\man{M}$ and
    $\{f(x_i)\}\subset\partial\imath_{d'}\man{M}$ we know that $q\in\partial\imath_{\imath_{d'}}\man{M}$.
    We shall show that $q=f(x)$. Since $\{x_i\}$ converges to $x$ uniquely, we may choose a sequence of
        open
        neighbourhoods, $V_i$, so that $x_i\in V_i$, for all
        $j\neq i$, $x_j\not\in V_i$ and for all $i, j,\ i\not = j$,
        $\overline{V_i}\cap\overline{V_j}=\EmptySet$.  
        Let
        $\{\imath_d y_j^i\}\subset\phi\man{M}$ be a sequence that
        converges uniquely to $x_i$ and is such that for all $j$,
        $\imath_d y^i_j\in V_i$. Let
        $\seq s=\bigcup_{i,j}\{\imath_d y_j^i\}$ be a new sequence formed
        from the union of the $\{\imath_d y_j^i\}$'s.  Since $\{f(x_i)\}$ has
        $q$ as a limit point we know that $\{f(\seq s)\}$ must also
        have $q$ as a limit point.  From paragraph \textbf{1} 
        we know that $f^{-1}(q)$ must be a limit point of
        $\seq s$. By construction this implies that $f^{-1}(q)$
        is either equal to $x_i$ for some $i$ or equal to $x$.

        \textbf{4} If $f^{-1}(q)=x$ then we are done, so suppose that there
        exists $l$ so that $q=f(x_l)$.  Since $q$ is a limit point of $\{f(x_i)\}$
        we can choose a subsequence, $\seq q=\{q_k\}$ of
        $\{f(\seq s)\}$ so that $q_r\in\{f(\imath_d y_j^r)\}$ and $\seq q$ uniquely converges to
        $q$. From paragraph \textbf{1} we know that $\{f^{-1}(q_r)\}$ must have
        $f^{-1}(q)=x_l$ as a unique limit point. This implies that
        $\{f^{-1}(q_r)\}\cap V_l$ must be infinite. But by
        construction we know that for all $r\not =l$,
        $y^r_j\not\in V_l$, and since $q_r=y^r_{j}$ for some $j$
        we know that $\{f^{-1}(q_r)\}\cap V_l$ is either empty or contains
        only the element $f^{-1}(q_l)$.  Therefore we have a contradiction
        and $q=f(x)$ as required. Thus for all sequences $\seq s$ lying in $\bound[\imath_d]{M}$ with the
        unique limit point $x$ we know that $f(\seq s)$ has the unique limit point $f(x)$ as required. Since the
        argument of the last three paragraphs applies to $f^{-1}$ we also know that 
        for all sequences $\seq s\subset\imath_{d'}\man{M}$ converging to $x\in\bound[\imath_{d'}]{M}$ that
        $f^{-1}(\seq s)\to f^{-1}(x)$.

    \textbf{5} Fourth, suppose that $\seq x$ is a sequence in $\man{M}^d$ so that $\seq x\to x\in\partial\imath_d\man{M}$. Suppose that
    $\seq x\cap\partial\imath_d\man{M}$ is finite and let $\seq w$ be the sequence given by $\seq x\cap\imath_d{\man{M}}$.  Then 
    from above we know that $f(\seq w)\to f(x)$ uniquely.  Since $\seq x-\seq w$ is finite we also know that $f(\seq x)$ must converge
    uniquely to $f(x)$.  We can use the same technique to show that $f(\seq x)\to f(x)$ uniquely if $\seq x\cap\imath_d\man{M}$ is finite. 
    So suppose that $\seq x\cap\imath_d\man{M}$ and $\seq x\cap\partial\imath_d\man{M}$ are infinite and let $\seq u=\seq x\cap\imath_d\man{M}$ and
    $\seq v=\seq x\cap\partial\imath_d\man{M}$.  From above we know that $f(\seq u)\to f(x)$ uniquely and that $f(\seq v)\to f(x)$ uniquely, therefore
    $f(\seq x)=f(\seq u)\cup f(\seq v)$ is such that $f(\seq x)\to f(x)$ uniquely. 

    The continuity of $f^{-1}$ follows by symmetry.

    $\Leftarrow$

    Let $\{x_i\}\in\Cau[d]$. Then $[\{x_i\}]_{d}\in\man{M}^d$
    and the sequence $\{\imath_d(x_i)\}_i$ converges uniquely to
    $[\{x_i\}]_d$.  Since
    $f(\imath_d(x_i))=\imath_{d'}(x_i)$ and as $f$ is a
    homeomorphism we know that $\{\imath_{d'}(x_i)\}_i$
    must uniquely converge to $[\{x_i\}]_{d'}$. Thus $\{x_i\}\in\Cau[d']$.

    By symmetry we can conclude that $\Cau[d]=\Cau[d']$ and that
    $d\simeq d'$ as required.
  \end{proof}

    This gives a useful corollary.

  \begin{Cor}\label{CorSusan}
    Let $d\in\Dis$, then there exists a
    homeomorphism $f:\man{M}^d\to\overline{\phi(\man{M})}$ so that
    $f\imath_d=\phi$, where $\phi\in E(d)$.
  \end{Cor}
  \begin{proof}
    Since $\phi\in E(d)$ there exists $d_\phi$, a complete distance on
    $\man{M}_\phi$, so that
    $d_\phi|_{\phi(\man{M})\times\phi(\man{M})}\simeq d$.
    Let $d'=d_\phi|_{\phi(\man{M})\times\phi(\man{M})}$.  From proposition
    \ref{DisEm:Prop.DistanceEquivalenceStatmentsNotInDis} we know
    that there exists a homeomorphism $h:\man{M}^d\to\man{M}^{d'}$
    where $h\imath_d=\imath_{d'}$.

    As $d_\phi$ is complete we know that $\overline{\phi(\man{M})}$ is
    homeomorphic to $\man{M}^{d'}$, and by the universality of the Cauchy completion we know that there exists a
    homeomorphism
    $g:\man{M}^{d'}\to\overline{\phi(\man{M})}$, so that $g\imath_{d'}=\phi$.

    Let $f:\man{M}^{d}\to\overline{\phi(\man{M})}$ be defined by
    $f=gh$.  Then $f$ is a homeomorphism and $f\imath_d=gh\imath_d=g\imath_{d'}=\phi$ as
    required.
  \end{proof}

  \begin{Prop}\label{DisEm:Prop.DistanceEquivalenceStatments}
    Let $d, d'\in\Dis$, then $d\simeq d'$
    if and only if $E(d)=E(d')$.
  \end{Prop}
  \begin{proof}
    $\Rightarrow$
    Let $\phi:\man{M}\to\man{M}_\phi$ be an element of $E(d)$,
    then there exists a complete distance
    $d_\phi:\man{M}_\phi\times\man{M}_\phi\to\mathbb{R}$ so that
    $d_\phi|_{\phi(\man{M})\times\phi(\man{M})}\simeq d$. Since
    $d\simeq d'$ and $\simeq$ is an equivalence relation, we
    can see that $d_\phi|_{\phi(\man{M})\times\phi(\man{M})}\simeq
    d'$. Hence $\phi \in E(d')$.
    By symmetry we can conclude that $E(d)=E(d')$.

    $\Leftarrow$
    Since $E(d)=E(d')$ we can choose $\phi:\man{M}\to\man{M}_\phi$
    an envelopment so that there exists two complete distances
    \[
    d_\phi:\man{M}_\phi\times\man{M}_\phi\to\mathbb{R}
    \]
    and
    \[
    d'_\phi:\man{M}_\phi\times\man{M}_\phi\to\mathbb{R}
    \]
    so that
    \[
    d_\phi|_{\phi(\man{M})\times\phi(\man{M})}\simeq d
    \]
    and
    \[
    d'_\phi|_{\phi(\man{M})\times\phi(\man{M})}\simeq d'.
    \]

    Now let $\seq u$ be a sequence in $\phi(\man{M})$.  If $\seq u$
    is Cauchy with respect to $d_\phi$ then there must exist
    $u\in\overline{\phi(\man{M})}$ so that $\seq u\to u$.  Since
    $\seq u\to u$ we can see that $\seq u$ must also be Cauchy
    with respect to $d'_\phi$.  This argument implies that
    $\Cau[d_\phi|_{\phi(\man{M})\times\phi(\man{M})}]\subset\Cau[d'_\phi|_{\phi(\man{M})\times\phi(\man{M})}]$.
    The same argument can be used, in the reverse direction, to
    show that
    $\Cau[d'_\phi|_{\phi(\man{M})\times\phi(\man{M})}]\subset\Cau[d_\phi|_{\phi(\man{M})\times\phi(\man{M})}]$.
    Hence
    \[
    d_\phi|_{\phi(\man{M})\times\phi(\man{M})} \simeq
    d'_\phi|_{\phi(\man{M})\times\phi(\man{M})}
    \]
    and since $\simeq$ is an equivalence relation we can conclude that
    \[
    d\simeq d'
    \]
    as required.

  \end{proof}

    These results, once combined, can give us the following
    corollary which is a converse of corollary \ref{CorSusan}.

  \begin{Cor}\label{CorSusanConv}
    Let $d$ be a distance on \man{M} and $\phi:\man{M}\to\man{M}_\phi$
    a non-trivial envelopment.  If there exists a
    homeomorphism $f:\man{M}^d\to\overline{\phi(\man{M})}$ so that
    $f\imath_d=\phi$, then $\phi\in E(d)$ and $d\in\Dis$.
  \end{Cor}
  \begin{proof}
    Let $d_\phi$ be a complete distance on $\man{M}_\phi$.  Let
    $d'=d_\phi|_{\phi(\man{M})\times\phi(\man{M})}$ then, by
    definition, $\phi\in E(d')$ and so $d'\in\Dis$. From corollary
    \ref{CorSusan} there must exist
    a homeomorphism $h:\man{M}^{d'}\to\overline{\phi(\man{M})}$
    so that $h\imath_{d'}=\phi$.

    Let $g:\man{M}^d\to\man{M}^{d'}$ be defined by $g=h^{-1}f$.
    Then $g$ is a homeomorphism and $g\imath_d=h^{-1}f\imath_{d}=h^{-1}\phi=\imath_{d'}$.
    Hence, by proposition
    \ref{DisEm:Prop.DistanceEquivalenceStatmentsNotInDis}, we can
    see that $d\simeq d'$.  Therefore, from definition
    \ref{DisEmb:Def.EnvelopebleDistance}, $\phi\in E(d)$ and $d\in\Dis$ as
    required.
  \end{proof}

  \subsection{The alternative proof of proposition \ref{DisEm:Prop.DistanceEquivalenceStatmentsNotInDis}}\label{distance Cauchy}
    
    Just as we gave an alternative proof of proposition \ref{DisEm:Prop.ConditionsForEquivalenceOfEmbeddings} using
    the extension theorem \ref{newTHerem} for Cauchy continuous functions, we now give an alternative proof of
    proposition \ref{DisEm:Prop.DistanceEquivalenceStatmentsNotInDis}.

  \begin{proof}[Alternative proof of proposition \ref{DisEm:Prop.DistanceEquivalenceStatmentsNotInDis}]
    Suppose that $d\simeq d'$. Give $\man{M}$ the Cauchy structure, $C_d$, induced by $d$ and give $\man{M}^{d'}$ the Cauchy structure,
    $C_{d'}$, given by 
        $d'$.  That is, $\seq s\in C_d$ ($\seq s\in C_{d'}$) if and only if $\seq s$ is Cauchy with
        respect to $d$ ($d'^*$). Note that
        $\seq s\in C_d$ implies $\seq s\subset\man{M}$ while $\seq s\in C_{d'}$ implies
        that $\seq s\subset \man{M}^{d'}$. Define $f:\man{M}\to\man{M}^{d'}$ by $f(x)=\imath_{d'}(x)$. Essentially we will
        show that $\imath_{d'}$ is Cauchy continuous with respect to the Cauchy structure induced by $d$.

        Let $\seq s\in C_d$. If $\seq s$ converges to $s$ in $\man{M}$ then, by the
        continuity of $\imath_{d'}$, we know that $f(\seq s)=\imath_{d'}(\seq s)$ must converge to $\imath_{d'}(s)$. That is, 
        $f(\seq s)$ is Cauchy with respect to $d'^*$.   
        So suppose that $\seq s\in C_d$ but does not have a limit point in $\man{M}$. 
        We then know that $\seq s\in C(d)$. As $C(d)=C(d')$ we can immediately conclude that
        $f(\seq s)\in C_{d'}$ and hence $f$ is Cauchy continuous.
        Thus, by theorem \ref{newTHerem}, there exists a unique extension of $f$, $\hat{f}:\man{M}^d\to\man{M}^{d'}$ so that $\hat f\imath_d=\imath_{d'}$. 
        Moreover,
        by the uniqueness of $\hat f$ 
        and, as the argument above is symmetric in $d$ and $d'$, we can see that $\hat f$ must be a homeomorphism.
        
        Suppose that we have such an $f$. Let $\seq s\in C(d)$. As $f$ is continuous, 
        $f\imath_d(\seq s)$ must be Cauchy with respect to $d'^*$ and cannot have a limit
        point in $\imath_{d'}(\man{M})$. Therefore $\seq s\in C(d')$. By symmetry we can conclude that $C(d)=C(d')$ as required.
  \end{proof}

\section{A correspondence between the equivalence classes}\label{sec.correspondence}

    We begin with some definitions.

  \begin{Def}
    Let $\Phi$ be the set of all envelopments of $\man{M}$.
  \end{Def}

  \begin{Def}
    Let $[\phi]$ denote the equivalence class of $\phi\in\Phi$
    under the equivalence relation $\simeq$.
  \end{Def}

  \begin{Def}
    Let $[d]$ denote the equivalence class of $d\in \Dis$ under the
    equivalence relation $\simeq$.
  \end{Def}

    We will show that there is a one-to-one correspondence
    between $\frac{\Dis}{\simeq}$ and $\frac{\Phi}{\simeq}$ by
    constructing two functions which are inverses of each other\footnote{In principle this correspondence can be derived using theorem \ref{newTHerem} by comparing 
    the Cauchy structures induced by an envelopment $\phi:\man{M}\to\man{M}_\phi$ and the Cauchy structure induced by a 
    complete distance on $\man{M}_\phi$. We choose not to follow this route, opting instead to give explicit definitions
    for the various maps. This does not entail additional complications in the proofs.}.
    First, we give the function from $\frac{\Dis}{\simeq}$ to
    $\frac{\Phi}{\simeq}$.

  \begin{Def}
    For each $d\in \Dis$ choose $\phi\in E(d)$. We shall denote this
    chosen
    envelopment by $\phi_d$.
    Define $I:\frac{\Dis}{\simeq}\to\frac{\Phi}{\simeq}$ by letting
    $I([d])=[\phi_d]$.
  \end{Def}

  \begin{Lem}\label{DisEm:Lem.IWellDefined}
    The function $I$ is well defined.
  \end{Lem}
  \begin{proof}
    Let $d,d'\in \Dis$ such that $d\simeq d'$.  Since
    $\phi_d\in E(d)$ and $\phi_{d'}\in E(d')$ there exist
    two homeomorphisms $h:\man{M}^d\to\overline{\phi_d(\man{M})}$
    and $f:\man{M}^{d'}\to\overline{\phi_{d'}(\man{M})}$. Also,
    as $d\simeq d'$ there exists a homeomorphism
    $g:\man{M}^d\to\man{M}^{d'}$. Therefore
    $fgh^{-1}:\overline{\phi_d(\man{M})}\to\overline{\phi_{d'}(\man{M})}$
    is a homeomorphism. By construction,
    $fgh^{-1}\phi_d=\phi_{d'}$, and therefore, by proposition
    \ref{DisEm:Prop.ConditionsForEquivalenceOfEmbeddings},
    $\phi_d\simeq \phi_{d'}$.
    Thus $I([d])=I([d'])$ and $I$ must be well defined.
  \end{proof}

    Now we construct the function from $\frac{\Phi}{\simeq}$ to
    $\frac{\Dis}{\simeq}$.

  \begin{Def}\label{JDef}
    Let $\psi\in \Phi$ and choose
    $d:\man{M}_\psi\times\man{M}_\psi\to\mathbb{R}$ to be a complete distance. Let
    $d_\psi:\man{M}\times\man{M}\to\mathbb{R}$ be defined by
    $d_\psi(x,y)=d(\psi(x),\psi(y))$, for all $x,y\in\man{M}$.
    Note that, by construction, $\psi\in
    E(d_\psi)$ and thus $d_\psi\in\Dis$.
    Define $J:\frac{\Phi}{\simeq}\to\frac{\Dis}{\simeq}$ by
    $J([\psi])=[d_\psi]$.
  \end{Def}

  \begin{Lem}\label{DisEm:Lem.JWellDefined}
    The function $J$ is well defined.
  \end{Lem}
  \begin{proof}
    Let $\psi,\phi\in\Phi$ such that $\psi\simeq\phi$, then there
    exists a homeomorphism
    $g:\overline{\psi(\man{M})}\to\overline{\phi(\man{M})}$. Also,
    since $\psi\in E(d_\psi)$ and $\phi\in E(d_\phi)$ there exist
    homeomorphisms $h:\man{M}^{d_\phi}\to\overline{\phi(\man{M})}$,
    $f:\man{M}^{d_\psi}\to\overline{\psi(\man{M})}$. Hence we have
    the homeomorphism
    $f^{-1}g^{-1}h:\man{M}^{d_\phi}\to\man{M}^{d_\psi}$. By
    construction,
    $f^{-1}g^{-1}h\imath_{d_\phi}=\imath_{d_\psi}$,
    implying $d_\phi\simeq d_\psi$, by proposition
    \ref{DisEm:Prop.DistanceEquivalenceStatmentsNotInDis}.
    Therefore $J([\psi])=J([\phi])$
    so that $J$ is well defined.
  \end{proof}

    As presented $I$ and $J$ are dependent on a choice of an
    envelopment and a distance, respectively.  It turns out that this
    choice is immaterial.

  \begin{Lem}
    For each $d\in\Dis$ the function $I$ is independent of the choice of $\phi_d$.
  \end{Lem}
  \begin{proof}
    Let $d\in\Dis$, then in order to prove the result
    we need to show that if $\psi\in E(d)$ then $\psi\simeq \phi_d$.

    As $\psi\in E(d)$ we know, from corollary \ref{CorSusan}, that
    there exists a homeomorphism $f:\man{M}^d\to\overline{\psi(\man{M})}$
    so that $f\imath_d=\psi$. Likewise, there must also exist a
    homeomorphism $g:\man{M}^d\to\overline{\phi_d(\man{M})}$
    so that $g\imath_d=\phi_d$.  Let $h=fg^{-1}$, so that
    $h:\overline{\phi_d(\man{M})}\to \overline{\psi(\man{M})}$ is a homeomorphism.  Also we
    know that $h\phi_d=fg^{-1}\phi_d=f\imath_d=\psi$ so that, by
    proposition
    \ref{DisEm:Prop.ConditionsForEquivalenceOfEmbeddings}, we can
    see that $\psi\simeq \phi_d$ as required.
  \end{proof}

  \begin{Lem}
    For each $\phi\in\Phi$ the function $J$ is independent of the choice of $d_\phi$.
  \end{Lem}
  \begin{proof}
    Let $\phi:\man{M}\to\man{M}_\phi$ be an element of
    $\Phi$ then, in order to prove our result, we need
    to show that for any two complete distances
    $d:\man{M}_\phi\times\man{M}_\phi\to\mathbb{R}$ and
    $d':\man{M}_\phi\times\man{M}_\phi\to\mathbb{R}$ the induced
    distances $d_\phi:\man{M}\times\man{M}\to\mathbb{R}$ and
    $d'_\phi:\man{M}\times\man{M}\to\mathbb{R}$, given by
    $d_\phi(x,y)=d(\phi(x),\phi(y))$ and
    $d'_\phi(x,y)=d'(\phi(x),\phi(y))$ (for all $x,y\in\man{M}$)
    are equivalent. That is, we
    need to show that $\Cau[d_\phi]=\Cau[d'_\phi]$.

    Let $\seq s\in\Cau[d_\phi]$ then $\seq s$ is Cauchy with
    respect to $d_\phi$ and  since $d$ is complete, by construction,
    there must
    exist $p\in\partial(\phi(\man{M}))$ so that $\phi(\seq s)\to p $
    uniquely.  This implies, however, that $\phi(\seq s)$ will be Cauchy
    with respect to any distance on $\man{M}_\phi$ and therefore
    $\phi(\seq s)$ is Cauchy with respect to $d'$.  Hence, by
    construction, $\seq s$ is Cauchy with respect to $d'_\phi$.
    Therefore $\Cau[d_\phi]\subset\Cau[d'_\phi]$.

    Similarly we can see that
    $\Cau[d'_\phi]\subset\Cau[d_\phi]$ and thus $d_\phi\simeq
    d'_\phi$ as required.
  \end{proof}
  
  Now we present our main results.

  \begin{Lem}\label{DisEm:Lem.JI=1}
    Let $d\in \Dis$, then $JI([d])=[d]$.
  \end{Lem}
  \begin{proof}
    Let $I([d])=[\phi_d]$, then $\phi_d\in E(d)$ so there exists a
    homeomorphism
    $f:\man{M}^{d}\to\overline{\phi_d(\man{M})}$. Since
    $J([\phi_d])=[d_{\phi_d}]$ we know that there exists a
    homeomorphism
    $g:\man{M}^{d_{\phi_d}}\to\overline{\phi_d(\man{M})}$. Thus
    $f^{-1}g:\man{M}^{d_{\phi_d}}\to\man{M}^d$ is a homeomorphism.
    By consulting the definitions we can see that
    $f^{-1}g\imath_{d_{\phi_d}}=\imath_d$
    and therefore $d_{\phi_d}\simeq d$.  Hence
    $JI([d])=[d_{\phi_d}]=[d]$.
  \end{proof}

  \begin{Lem}\label{DisEm:Lem.IJ=1}
    Let $\psi\in\Phi$, then $IJ([\psi])=[\psi]$.
  \end{Lem}
  \begin{proof}
    Let $J([\psi])=[d_\psi]$ then there exists a homeomorphism
    $f:\man{M}^{d_\psi}\to\overline{\psi(\man{M})}$. Let
    $I([d_\psi])=[\phi_{d_\psi}]$ then there exists a homeomorphism
    $g:\man{M}^{d_\psi}\to\overline{\phi_{d_\psi}(\man{M})}$. Since
    $fg^{-1}:\overline{\phi_{d_\psi}(\man{M})}\to\overline{\psi(\man{M})}$
    is a homeomorphism and as
    $fg^{-1}\phi_{d_\psi}=\psi$ we conclude that $\psi\simeq \phi_{d_\psi}$. Therefore
    $IJ([\psi])=[\phi_{d_\psi}]=[\psi]$.
  \end{proof}  

  \begin{Thm}
    The function $I:\frac{\Dis}{\simeq}\to\frac{\Phi}{\simeq}$ is a
    bijective function with inverse $J$. That is, the sets
    $\frac{\Dis}{\simeq}$ and $\frac{\Phi}{\simeq}$ are in one-to-one
    correspondence with each other.
  \end{Thm}
  \begin{proof}
    This follows from lemmas \ref{DisEm:Lem.JI=1} and
    \ref{DisEm:Lem.IJ=1}.
  \end{proof}
  
  This theorem shows us that any information that can be
    extracted from $\frac{\Phi}{\simeq}$ (e.g., the $a$-boundary) can
    also be extracted from $\frac{\Dis}{\simeq}$.  

\section{Demonstration of correspondence}\label{sec.demonstration}

  To illustrate how this correspondence can be used we show how the `covering' and `in contact' relations between
  boundary points of the two maximal extensions
  of the Misner space-time, for $t>0$, can be constructed using envelopments or their corresponding distances.
  
  \subsection{The Misner space-time}
  
  The upper-half Misner space-time is the space-time with manifold $\man{M}=\mathbb{R}^+\times S^1$ with metric, in the coordinates $t$ and
  $\psi$, $0<t<\infty$, $0\leq\psi<2\pi$, given by
  \[
    g=\begin{pmatrix}
      \frac{-1}{t} &0\\
      0& t
    \end{pmatrix}.
  \]

  It is well known \cite{HawkingEllis1973} that there exist two maximal extensions of $\man{M}$. Let $\man{M}_{1}=\mathbb{R}\times S^1$
  with coordinates $t$, $\psi_1$, $t\in\mathbb{R}$, $0\leq\psi_1<2\pi$ and metric
  \[
    g_1=\begin{pmatrix}
      0&1\\
      1&t
    \end{pmatrix}.
  \]
  Let $\man{M}_2=\mathbb{R}\times S^1$ with coordinates $t$, $\psi_2$, $t\in\mathbb{R}$, $0\leq\psi_2<2\pi$ and metric
  \[
    g_2=\begin{pmatrix}
      0&-1\\
      -1&t
    \end{pmatrix}.    
  \]
  There are two envelopments $\phi_1:\man{M}\to\man{M}_1$ and $\phi_2:\man{M}\to\man{M}_2$ given by
  \begin{align*}
    \phi_1(t,\psi)&=(t,\psi-\log t \mod 2\pi)\\
    \phi_2(t,\psi)&=(t,\psi+\log t \mod 2\pi).
  \end{align*}
  In both cases $\phi_i(\man{M})$ is isometric to $\man{M}$, $i=1,2$.
  These manifolds and the maps between them provide a wealth of counter-examples for various conjectures in General Relativity.
  
  \subsection{A useful sequence}
  
    For the purposes of this section it is important to construct a sequence in $\mathbb{R}^+$ having certain properties. 
    Let $r\in\mathbb{Q}\cap(0,2\pi)$.
    Let $t_n=nr$, for all $n\in\mathbb{N}$. Since $r$ is rational we know that the
    set $\{t_n \mod 2\pi\}$ is dense in $[0,2\pi)$. Let $s_n=\exp(-t_n)$.

    \subsubsection{Construction of the `in contact' and `covering' relation using envelopments}
      
      Choose $(0,\psi_1)\in\partial(\phi_1(\man{M}))$ and $(0,\psi_2)\in\partial(\phi_2(\man{M}))$. We will show that these two arbitrary points
      on the boundaries
      are in contact and that neither covers the other. 
      Take the sequence $\seq s=\{(\sqrt{s_n},\psi_1+\log\sqrt{s_n} \mod 2\pi)\}$ 
      in $\man{M}$ so that $\phi_1(\seq s)=\{(\sqrt{s_n},\psi_1)\}$ clearly converges to
      $(0,\psi_1)$. We can calculate that $\phi_2(\seq s)=\{(\sqrt{s_n},\psi_1+2\log\sqrt{s_n} \mod 2\pi)\}=\{(\sqrt{s_n},\psi_1-t_n \mod 2\pi)\}$. 
      Since $\{t_n\mod 2\pi\}$ is
      dense in $[0,2\pi)$ there exists a subsequence $\{u_n\}\subset\{t_n\}$ so that $\{u_n \mod 2\pi\}$ converges to $\psi_2-\varphi \mod 2\pi$,
      where $\varphi\in[0,2\pi)$ is arbitrary. Let $q_n=\exp(-u_n)$ and
      $\seq q=\{(\sqrt{q_n},\psi_1+\log\sqrt{q_n} \mod 2\pi)\}$. 
      It is clear that $\seq q$ is a subsequence of $\seq s$ so that $\phi_1(\seq q)$ must converge
      to $(0,\psi_1)$. By construction we also have the following:
      \begin{align*}
        \lim_{n\to\infty}\phi_2(\seq q)&=\lim_{n\to\infty}\{(\sqrt{q_n},\psi_1+2\log\sqrt{q_n} \mod 2\pi)\}\\
          &=\lim_{n\to\infty}\{(\sqrt{q_n},\psi_1-{u_n} \mod 2\pi)\}\\
          &=(0,\psi_1-\psi_2+\varphi \mod 2\pi).
      \end{align*}
      Since $\varphi\in[0,2\pi)$ is arbitrary, the calculation above shows that $\phi_2(\seq s)$ has every point of $\partial(\phi_2(\man{M}))$ as
      a limit point. In particular $\phi_2(\seq q)$ converges to $(0,\psi_2)$ for $\varphi=2\psi_2-\psi_1 \mod 2\pi$. 
      By symmetry this is enough to prove that
      $(0,\psi_1)\in\partial(\phi_1(\man{M}))$ and $(0,\psi_2)\in\partial(\phi_2(\man{M}))$ 
      are in contact but that neither covers the other.
      
    \subsubsection{Construction of the `in contact' and `covering' relation using distances}
    
      Members of the equivalence classes of distances on $\man{M}$ induced by $\phi_1$ and $\phi_2$ can be calculated as follows.
      First we take two complete distances, $d_1$ and $d_2$, on $\man{M}_1$ and $\man{M}_2$. These are, respectively:
      \begin{align*}
        d_1\bigl((t,\psi_1),(t',\psi_1')\bigr)&=\sqrt{\left(t-t'\right)^2+\left(\psi_1-\psi_1'\right)^2}\\
        d_2\bigl((t,\psi_2),(t',\psi_2')\bigr)&=\sqrt{\left(t-t'\right)^2+\left(\psi_2-\psi_2'\right)^2}.
      \end{align*}
      Next we pull these back to $\man{M}$. In an abuse of notation we will also denote the pull backs by $d_1$ and $d_2$. The result is
      \begin{align*}
        d_1\bigl((t,\psi),(t',\psi')\bigr)&=\sqrt{\left(t-t'\right)^2+\bigl((\psi-\log t\mod 2\pi)-(\psi'-\log t'\mod 2\pi)\bigr)^2}\\
        d_2\bigl((t,\psi),(t',\psi')\bigr)&=\sqrt{\left(t-t'\right)^2+\bigl((\psi+\log t\mod 2\pi)-(\psi'+\log t'\mod 2\pi)\bigr)^2}.
      \end{align*}
      These are representatives of the equivalence classes of distances which correspond to the equivalence classes of $\phi_1$ and $\phi_2$.
      
      Let $\psi_1'\in[0,2\pi).$ We can check that the sequence $\seq s_1=\{(\sqrt{s_n},\psi_1'+\log\sqrt{s_n}\mod 2\pi)\}$ in $\man{M}$
      is Cauchy with respect to $d_1$:
      \begin{align*}
        d_1\bigl((\sqrt{s_n},\psi_1'&+\log\sqrt{s_n}\mod 2\pi),(\sqrt{s_m},\psi_1'+\log\sqrt{s_m}\mod 2\pi)\bigr)\\
        &=\sqrt{(\sqrt{s_n}-\sqrt{s_m})^2}\\
        &=\left|\sqrt{s_n}-\sqrt{s_m}\right|.
      \end{align*}
      As $\{s_n\}$ converges to $0$, we can take $n$ and $m$ large enough to make the distance above arbitrarily small. Therefore
      $\seq s_1$ is Cauchy with respect to $d_1$.
      From corollary \ref{CorSusan} and definition \ref{JDef} we can see that $[\{(\sqrt{s_n},\psi_1'+\log\sqrt{s_n}\mod 2\pi)\}]_{d_1}$
      corresponds to the boundary point $(0,\psi_1')$ in $\partial(\phi_1(\man{M}))$.
      Repeating this argument for $\phi_2$ we see that the sequence
      $\seq s_2=\{(\sqrt{s_n},\psi_2'-\log\sqrt{s_n}\mod 2\pi)\}$ is Cauchy with respect to $d_2$
      and corresponds to the boundary point $(0,\psi_2')$ in $\partial(\phi_2(\man{M}))$.
      
      We may now calculate the distance between these sequences with respect to $d_1$:
      \begin{multline*}
       d_1\bigl((\sqrt{s_n},\psi_1'+\log\sqrt{s_n}\mod 2\pi),(\sqrt{s_m},\psi_2'-\log\sqrt{s_m}\mod 2\pi)\bigr)\\
        =\sqrt{(\sqrt{s_n}-\sqrt{s_m})^2+(\psi_1'-\psi_2'-{t_m} \mod 2\pi)^2}.
      \end{multline*}
      Since $\{t_m\mod 2\pi\}$ is dense in $[0,2\pi)$ 
      we can immediately see that for all $\epsilon<0$ there exists $n_0,m_0\in\mathbb{N}$ so that 
      $$d_1\bigl((\sqrt{s_{n_0}},\psi_1'+\log\sqrt{s_{n_0}}\mod 2\pi),(\sqrt{s_{m_0}},\psi_2'-\log\sqrt{s_{m_0}}\mod 2\pi)\bigr)<\epsilon$$
      but that there does not exist $N\in\mathbb{N}$ so that for all $n,m>N$ this is true.
      
      If the points $[\seq s_1]_{d_1}$ and $[\seq s_2]_{d_2}$ covered each other
      then the distance, above, would have to limit to zero as $n,m\to\infty$. We have just demonstrated
      that the distance is not zero in the limit and therefore neither point covers the other. Moreover 
      there exists $\seq q\subset \seq s_2$ so that $[\{(\sqrt{s_n},\psi_1'+\log\sqrt{s_n}\mod 2\pi)\}]_{d_1}=[\seq q]_{d_1}$ and
      therefore the points $[\seq s_1]_{d_1}$ and $[\seq s_2]_{d_2}$ are in contact.
      Thus, under the correspondence, we know that
      $(0,\psi_1')\in\partial(\phi_1(\man{M}))$ and $(0,\psi_2')\in\partial(\phi_2(\man{M}))$ 
      are in contact but neither covers the other.

  \subsection{Discussion}

    The construction of the relations via distances and envelopments are of a similar complexity but use very different techniques. It
    is this difference that this paper advocates. One now has a choice of techniques for working with the Abstract Boundary. We note that
    until an envelopment independent definition of $D(\man{M})$ is given there is still some level of dependence on envelopments.
    This came through above via the use of pull backs to define our distances. This problem is therefore of a pressing nature. Unfortunately this
    is a very difficult problem since it requires a characterization of the distances on the manifold which correspond to envelopments.

    That every boundary point in $\partial(\phi_1(\man{M}))$ is `in contact' with every boundary point in $\partial(\phi_2(\man{M}))$ but no two
    points cover each other expresses the fact that the boundary $\partial(\phi_1(\man{M}))$ is `smeared' over the boundary 
    $\partial(\phi_2(\man{M}))$ and vice versa. The sets $\sigma_{\phi_1}$ and $\sigma_{\phi_2}$ are therefore each partial cross sections
    which contain the same boundary information expressed in very different ways. The boundary information 
    associated with a point of $\partial(\phi_1(\man{M}))$ is spread over every boundary point in $\partial(\phi_2(\man{M}))$.
    This, very odd, behaviour gives an example of how the Abstract Boundary copes with multiple maximal envelopments.

\section{Conclusions}
 
    We have defined equivalence relations on the set of all envelopments of a manifold and a subset of the set
    of all distances on a manifold.  The resulting sets of equivalence classes were then shown to be in one-to-one correspondence
    with each other; hence they are `the same.'  Since the Abstract Boundary can be constructed from $\frac{\Phi}{\simeq}$, we can conclude that 
    it is possible to construct the Abstract Boundary using $\frac{\Dis}{\simeq}$ instead. In a following paper we will
    show how this can be done.  Therefore, instead of thinking about boundary points of a particular envelopment of the manifold, we
    can now think about collections of Cauchy sequences with
    respect to some distance on the manifold.
    
    The Abstract Boundary has already proven to be a very useful, intuitive construction (see \cite{AshleyScott2003}).  By showing how envelopments can be
    replaced by distances, the construction can now be applied in new ways.  Note, however, 
    that to define the set $\Dis$ we had to refer to envelopments, via the sets $E(d)$.  So, while
    we have presented an alternative way to view the Abstract Boundary, in order to fully dissociate the two approaches (distances vs.
    envelopments) we still need to find a definition for $\Dis$ that does not in any way rely on envelopments.  Current research is pursuing
    this goal.  
    
    In summary, we
    have shown that the structure of the edge of a space-time can
    be deduced from knowledge of a distance defined purely on the space-time itself. Given the successes of the Abstract Boundary, this 
    thereby provides new tools when 
    working with the edge of a space-time.
    
      By giving the
    relationship between $\frac{\Phi}{\simeq}$ and $\frac{\Dis}{\simeq}$ we have demonstrated how a fundamental building block of the Abstract Boundary can be 
    replaced when required.
    
\section{Acknowledgements}
  The authors would like to thank our reviewer whose comments improved the clarity of the paper.


\begin{thebibliography}{10}

\bibitem{Ashley2002a}
M.~J.~S.~L. Ashley.
\newblock {\em Singularity Theorems and the Abstract Boundary Construction}.
\newblock PhD thesis, Department of Physics, The Australian National
  University, 2002.
\newblock Located at
  {http://thesis.anu.edu.au/public/adt-ANU20050209.165310/index.html}.

\bibitem{AshleyScott2003}
M.~J.~S.~L. Ashley and S.~M. Scott.
\newblock Curvature singularities and abstract boundary singularity theorems
  for space-time.
\newblock In {\em Recent advances in Riemannian and Lorentzian geometries
  (Baltimore, MD, 2003)}, volume 337 of {\em Contem. Math.}, pages 9--19.
  Providence, RI, 2003.

\bibitem{BeemEhrlichEasley1996}
J.~K. Beem, P.~E. Ehrlich, and K.~L. Easley.
\newblock {\em Global {L}orentzian {G}eometry}, volume 202 of {\em Pure and
  Applied Mathematics: A Series of Monographs and Textbooks}.
\newblock Marcel Dekker, Inc., 1996.

\bibitem{0264-9381-22-21-009}
A.~Garc\'ia-Parrado and M.~S\'anchez.
\newblock Further properties of causal relationship: causal structure
  stability, new criteria for isocausality and counterexamples.
\newblock {\em Class. Quant. Grav.}, 22(21):4589--4619, 2005.

\bibitem{GarciaParrado:2002xt}
A.~Garc\'ia-Parrado and J.~M.~M. Senovilla.
\newblock {Causal relationship: A new tool for the causal characterization of
  Lorentzian manifolds}.
\newblock {\em Class. Quant. Grav.}, 20:625--664, 2003.

\bibitem{0264-9381-22-9-R01}
A.~Garc\'{i}a-Parrado and J.~M.~M. Senovilla.
\newblock Causal structures and causal boundaries.
\newblock {\em Class. Quant. Grav.}, 22(9):R1--R84, 2005.

\bibitem{Geroch1968a}
R.~Geroch.
\newblock Local characterization of singularities in general relativity.
\newblock {\em J. Math. Phys.}, 9:450--465, 1968.

\bibitem{GerochPenroseKronheim1972}
R.~Geroch, R.~Penrose, and E.~H. Kronheimer.
\newblock Ideal points in space-time.
\newblock {\em Proc. R. Soc. Lond.}, A327(1571):545--567, 1972.

\bibitem{HawkingEllis1973}
S.~W. Hawking and G.~F.~R. Ellis.
\newblock {\em The Large Scale Structure of Space-Time}.
\newblock Cambridge University Press, 1973.

\bibitem{Flores2010Final}
{J. L. Flores}, {J. Herrera}, and {M. Sanchez}.
\newblock {On the final definition of the causal boundary and its relation with
  the conformal boundary}.
\newblock {Jan} 2010.
\newblock {/abs/1001.3270}.

\bibitem{LowenColebunders1989}
E.~Lowen-Colebunders.
\newblock {\em Function classes of cauchy continuous maps}.
\newblock Marcel Dekker, Inc., 1989.

\bibitem{Philpot2004}
L.~Philpot.
\newblock \textit{Singularity Constructions for Space-time}.
\newblock Honours Thesis, 2004.
\newblock {C}ontact the authors or the Department of Physics, The Australian
  National University, for a copy.

\bibitem{Reed1971}
E.~E. Reed.
\newblock Completions of uniform convergence spaces.
\newblock {\em Mathematische Annalen}, 194(2):83--108, 1971.

\bibitem{Sanchez2009e1744}
M.~S\'{a}nchez.
\newblock Causal boundaries and holography on wave type spacetimes.
\newblock {\em Nonlinear Analysis: Theory, Methods \& Applications},
  71(12):e1744 -- e1764, 2009.

\bibitem{Schmidt1971}
B.~G. Schmidt.
\newblock A new definition of singular points in general relativity.
\newblock {\em Gen. Rel. Grav.}, 1(3):269--280, 1971.

\bibitem{ScottSzekeres1994}
S.~M. Scott and P.~Szekeres.
\newblock The abstract boundary---a new approach to singularities of manifolds.
\newblock {\em J. Geom. Phys.}, 13(3):223--253, 1994.

\bibitem{Senovilla1998}
J.~M.~M. Senovilla.
\newblock Singularity theorems and their consequences.
\newblock {\em Gen. Rel. Grav.}, 30(5):701--848, 1998.

\bibitem{Whale2010}
B.~E. Whale.
\newblock {\em Foundations of and Applications for the Abstract Boundary
  Construction for Space-Time}.
\newblock PhD thesis, Department of Quantum Science, The Australian National
  University, 2010.
\newblock Located at
  {http://thesis.anu.edu.au/public/adt-ANU20101106.100700/index.html}.

\end{thebibliography}
\end{document}